\documentclass[envcountsame,envcountsect]{llncs}

\usepackage{amsfonts}
\usepackage{amsmath, amssymb}

\usepackage{latexsym}
\usepackage[pdftex]{graphicx}

\usepackage{todonotes}

\usepackage{tikz}
\usetikzlibrary{matrix,arrows,shapes,snakes,automata,backgrounds,positioning}

\newcommand{\uu}{\boldsymbol{u}}
\newcommand{\vv}{\boldsymbol{v}}

\newcommand{\size}[1]{|| #1 ||}

\newcommand{\defemph}[1]{\textbf{\textsl{#1}}}

\begin{document}
	\title{Ultimate Positivity is Decidable for Simple Linear
          Recurrence Sequences\thanks{The full version of this paper 
is available as~\cite{OuaknineW13b}.}}

	\author{Jo\"el Ouaknine \and James Worrell}
\institute{Department of Computer Science, University of Oxford, UK}
	
	\maketitle
	
\begin{abstract}          
We consider the decidability and complexity of the Ultimate Positivity
Problem, which asks whether all but finitely many terms of a given
rational linear recurrence sequence (LRS) are positive.  Using lower
bounds in Diophantine approximation concerning sums of $S$-units, we
show that for simple LRS (those whose characteristic polynomial has no
repeated roots) the Ultimate Positivity Problem is decidable in
polynomial space.  If we restrict to simple LRS of a fixed order then
we obtain a polynomial-time decision procedure.  As a complexity lower
bound we show that Ultimate Positivity for simple LRS is at least as
hard as the decision problem for the universal theory of the reals:
a problem that is known to lie between  \textbf{coNP} and
\textbf{PSPACE}.
\end{abstract}
	
\section{Introduction}
A \defemph{linear recurrence sequence (LRS)} is an infinite
sequence $\uu = \langle u_0, u_1, \ldots \rangle$ of rational
numbers satisfying a recurrence relation
\begin{equation}
\label{rec-rel}
u_{n+k} = a_1 u_{n+k-1} + a_2 u_{n+k-2} + \ldots + a_k u_{n} 
\end{equation} 
for all $n\geq 0$, where $a_1, a_2, \ldots, a_k$ are fixed rational
numbers with $a_k\neq 0$.  Such a sequence is determined by its
initial values $u_0, \ldots, u_{k-1}$ and the recurrence relation.
We say that the recurrence has \defemph{characteristic polynomial}
\[f(x) = x^k - a_1 x^{k-1} - \ldots - a_{k-1} x - a_k \, .\]
The least $k$ such that $\uu$ satisfies a recurrence of the form
(\ref{rec-rel}) is called the \defemph{order} of $\uu$.  If the
characteristic polynomial of this (unique) recurrence has no repeated
roots then we say that $\uu$ is \defemph{simple}. 

Given an LRS $\uu$ there are polynomials $p_1,\ldots,p_k
\in \mathbb{C}[x]$ such that
\[ u_n = p_1(n)\gamma_1^n + \ldots + p_k(n) \gamma_k^n \, , \]
where $\gamma_1,\ldots,\gamma_k$ are the roots of the characteristic
polynomial.  Moreover $\uu$ is simple if and only if it admits such a
representation in which each polynomial $p_i$ is a constant.  Simple
LRS are a natural and widely studied subclass of LRS whose analysis
nevertheless remains extremely
challenging~\cite{AV09,BOOK,ESS02,OW13:constructive-positivity}.

Motivated by questions in language theory and formal power series,
Rozenberg, Salomaa, and Soittola~\cite{RS94,Soi76} highlight the
following four decision problems concerning LRS\@.  Given an LRS
$\langle u_n \rangle_{n=0}^\infty$ (represented by a linear recurrence
and sequence of initial values):
\begin{enumerate}
\item Does $u_n=0$ for some $n$?
\item Does $u_n=0$ for infinitely many $n$?
\item Is $u_n \geq 0$ for all $n$?
\item Is $u_n \geq 0$ for all but finitely many $n$?
\end{enumerate}

Linear recurrence sequences are ubiquitous in mathematics and computer
science, and the above four problems (and assorted variants) arise in
a variety of settings; see~\cite{OW14:SODA} for references.  For
example, an LRS modelling population size is biologically meaningful
only if it never becomes negative.

Problem~1 is known as \defemph{Skolem's Problem}, after the
Skolem-Mahler-Lech Theorem~\cite{Lec53,Mah35,Sko34}, which
characterises the set $\{ n \in \mathbb{N} : u_n = 0\}$ of zeros of an
LRS $\uu$ as an ultimately periodic set.  The proof of the
Skolem-Mahler-Lech Theorem is non-effective, and the decidability of
Skolem's Problem is open.  Blondel and Tsitsiklis~\cite{BlondelT00}
remark that ``the present consensus among number theorists is that an
algorithm [for Skolem's Problem] should exist''.  However, so far
decidability is known only for LRS of order at most 4: a result due
independently to Vereschagin~\cite{Ver85} and Mignotte, Shorey, and
Tijdeman~\cite{MST84}.  At order 5 decidability is not known, even for
simple LRS~\cite{OW12}.  Decidability of Skolem's Problem is also
listed as an open problem and discussed at length by Tao~\cite[Section
  3.9]{Tao08}.  The problem can furthermore be seen as a
generalisation of the Orbit Problem, studied by Kannan and
Lipton~\cite[Section 5]{KL86}.

In contrast to the situation with Skolem's Problem,
Problem~2---hitting zero infinitely often---was shown to be decidable
for arbitrary LRS by Berstel and Mignotte~\cite{BM76}.
  
Problems 3 and 4 are respectively known as the \defemph{Positivity}
and \defemph{Ultimate Positivity} Problems.
The problems are stated as open
in~\cite{BG07,HHH06,LT09}, among others, while in~\cite{RS94} the
authors assert that the problems are ``generally conjectured [to be]
decidable''.  Decidability of Positivity entails decidability of
Skolem's Problem via a straightforward algebraic transformation of LRS
(which however does not preserve the order)~\cite{HHH06}.

Hitherto, all decidability results for Positivity and Ultimate
Positivity have been for low-order sequences.  The
paper~\cite{OW14:SODA} gives a detailed account of these results,
obtained over a period of time stretching back some 30 years, and
proves decidability of both problems for sequences of order at most 5.
It is moreover shown in~\cite{OW14:SODA} that obtaining
decidability for either Positivity or Ultimate Positivity at order 6
would necessarily entail major breakthroughs in Diophantine
approximation.

The main result of this paper is that the Ultimate Positivity Problem
for simple LRS of arbitrary order is decidable.  The restriction to
simple LRS allows us to circumvent the strong ``mathematical
hardness'' result for sequences of order 6 alluded to above.  However,
our decision procedure is non-constructive: given an ultimately
positive LRS $\langle u_n \rangle_{n=0}^\infty$, the procedure does
not compute a threshold $N$ such that $u_n \geq 0$ for all $n\geq N$.
Indeed the ability to compute such a threshold $N$ would immediately yield
an algorithm for the Positivity Problem for simple LRS since the
signs of $u_0,\ldots,u_{N-1}$ can be checked directly.  In turn
this would yield decidability of Skolem's Problem for simple LRS\@.
But Skolem's Problem is open for simple LRS of order 5, while (as
discussed below) Positivity for simple LRS is only known to be
decidable up to order 9.

The non-constructive aspect of our results arises from our use of
lower bounds in Diophantine approximation concerning sums of
\emph{$S$-units}.  These bounds were proven in~\cite{Evertse84,PS82}
using Schlickewei's ${p}$-adic generalisation of
Schmidt's Subspace Theorem (itself a far-reaching generalisation of
the Thue-Siegel-Roth Theorem), and therein applied to study the
asymptotic growth of LRS in absolute value.  By contrast,
in~\cite{OW13:constructive-positivity} we use Baker's Theorem on
linear forms in logarithms to show decidability of Positivity for
simple LRS of order at most 9.  Unfortunately, while Baker's Theorem
yields effective Diophantine-approximation lower bounds, it appears
only to be applicable to low-order LRS\@.  In particular, the analytic
and geometric arguments that are used
in~\cite{OW13:constructive-positivity} to bring Baker's Theorem to
bear (and which give that work a substantially different flavour to
the present paper) do not apply beyond order 9.

Relying on complexity bounds for the decision problem for first-order
formulas over the field of real numbers, we show that our procedure
for deciding Ultimate Positivity requires polynomial space in general
and polynomial time for LRS of each fixed order.  As a complexity
lower bound, we give a polynomial-time reduction of the decision
problem for the universal theory of the reals to both the Positivity
and Ultimate Positivity Problems for simple LRS\@.  The decision
problem for the universal theory of the reals is easily seen to be
\textbf{coNP}-hard and, from the work of Canny~\cite{Canny88}, is
contained in \textbf{PSPACE}.  Thus the complexity of the Ultimate
Positivity problem for simple LRS lies between \textbf{coNP} and
\textbf{PSPACE}.  Hitherto the best lower bound known for either
Positivity or Ultimate Positivity was
\textbf{coNP}-hardness~\cite{BellDJB10}.  

Full proofs of all results can be found in the long version of this 
paper~\cite{OuaknineW13b}.

\section{Background}
\label{sec:background}
\paragraph{Number Theory.}
A complex number $\alpha$ is \defemph{algebraic} if it is a root of a
univariate polynomial with integer coefficients. The \defemph{defining
  polynomial} of $\alpha$, denoted $p_{\alpha}$, is the unique
integer polynomial of least degree, whose coefficients have no common
factor, that has $\alpha$ as a root. The \defemph{degree} of $\alpha$
is the degree of $p_\alpha$, and the \defemph{height} of $\alpha$ is
the maximum absolute value of the coefficients of $p_{\alpha}$.  If
$p_\alpha$ is monic then we say that $\alpha$ is an \defemph{algebraic
  integer}.  

For computational purposes an algebraic number $\alpha$ can be
represented by a polynomial $f$ that has $\alpha$ as a root, together
with an approximation of $\alpha$ with rational real and imaginary
parts of sufficient accuracy to distinguish $\alpha$ from the other
roots of $f$~\cite{TUCS05}.  We denote by $\size{\alpha}$ the length
of this representation.\footnote{In general we denote by $\size{X}$
  the length of the binary representation of a given object $X$.}  It
can be shown that $\size{\alpha}$ is polynomial in the degree and
logarithm of the height of $\alpha$.  Given a univariate polynomial
$f$, it is moreover known how to obtain representations of each of its
roots in time polynomial in $\size{f}$.


A \defemph{number field} $K$ is a finite-dimensional extension of 
$\mathbb{Q}$.  The set of algebraic integers in $K$ forms a ring,
denoted $\mathcal{O}$.  Given two ideals $I,J$ in $\mathcal{O}$, the
product $IJ$ is the ideal generated by the elements $ab$,
where $a\in I$ and $b\in J$. An ideal $P$ of $\mathcal{O}$ is
\defemph{prime} if $ab\in P$ implies $a\in P$ or $b\in P$.  The
fundamental theorem of ideal theory states that any non-zero
ideal in $\mathcal{O}$ can be written as the product of prime ideals,
and the representation is unique if the order of the prime ideals is
ignored.

We will need the following
classical result of Dirichlet~\cite{FT93}.

\begin{theorem}[Dirichlet]
Let $P$ be the set of primes and $P_{a,b}$ the set of primes congruent
to $a$ mod $b$, where $\gcd(a,b)=1$.  Then
\[ \lim_{n\rightarrow \infty} \frac{|P_{a,b} \cap \{1,\ldots,n\}|}{|P\cap \{1,\ldots,n\}|} = \frac{1}{\varphi(b)} \, , \]
where $\varphi$ denotes Euler's totient function.
\label{thm:dirichlet}
\end{theorem}

\paragraph{Linear Recurrence Sequences.}
Let $\uu = \langle u_n \rangle_{n=0}^{\infty}$ be a sequence of
rational numbers satisfying the recurrence relation $u_{n+k} = a_1
u_{n+k-1} + \ldots + a_k u_{n}$.  We represent such an LRS as a
$2k$-tuple $(a_1, \ldots, a_k, u_0, \ldots, u_{k-1})$ of rational
numbers (encoded in binary).  Given an arbitrary representation of
$\uu$, we can compute the coefficients of the unique minimal-order
recurrence satisfied by $\uu$ in polynomial time by straightforward
linear algebra.  Henceforth we will always assume that an LRS is
presented in terms of its minimal-order recurrence.  By the
characteristic polynomial of an LRS we mean the characteristic
polynomial of the minimal-order recurrence.  The roots of this
polynomial are called the \defemph{characteristic roots}.  The
characteristic roots of maximum modulus are said to be
\defemph{dominant}.

It is well-known (see, e.g.,~\cite[Thm.~2]{BG07}) that if an LRS
$\uu$ has no real positive dominant
characteristic root then there are infinitely many $n$ such that $u_n
< 0$ and infinitely many $n$ such that $u_n > 0$.  Clearly such an LRS cannot
be ultimately positive.

Since the characteristic polynomial of $\uu$ has real coefficients,
its set of roots can be written in the form $\{\rho_1, \ldots,
\rho_{\ell}, \gamma_1, \overline{\gamma_1}, \ldots, \gamma_m,
\overline{\gamma_m}\}$, where each $\rho_i \in \mathbb{R}$.  If $\uu$
is simple then there are non-zero real algebraic constants $b_1,
\ldots, b_{\ell}$ and complex algebraic constants $c_1, \ldots, c_m$
such that, for all $n \geq 0$,
\begin{equation}
\label{eq:exp}
u_n = \sum_{i=1}^{\ell} b_i \rho_i^n + 
      \sum_{j=1}^m \left(c_j \gamma_j^n + 
                  \overline{c_j} \overline{\gamma_j}^n\right) \,.
\end{equation}

Conversely, a sequence $\uu$ that admits the representation
(\ref{eq:exp}) is a simple LRS over $\mathbb{R}$, with characteristic
roots among $\rho_1, \ldots, \rho_{\ell}, \gamma_1, \overline{\gamma_1},
\ldots, \gamma_m, \overline{\gamma_m}$.  Arbitrary LRS admit a more
general ``exponential-polynomial'' representation in which the
coefficients $b_i$ and $c_j$ are replaced by polynomials in $n$.

An LRS is said to be \defemph{non-degenerate} if it does not have two
distinct characteristic roots whose quotient is a root of unity.  A
non-degenerate LRS is either identically zero or only has finitely
many zeros.  The study of arbitrary LRS can effectively be reduced to
that of non-degenerate LRS using the following result
from~\cite{BOOK}.
\begin{proposition}
Let $\langle u_n \rangle_{n=0}^\infty$ be an LRS of order $k$ over
$\mathbb{Q}$.  There is a constant $M = 2^{O(k \sqrt{\log k})}$
such that each subsequence $\langle u_{Mn+l}
\rangle_{n=0}^\infty$ is non-degenerate for $0 \leq l < M$.
\label{prop:non-degenerate}
\end{proposition}

The constant $M$ in Proposition~\ref{prop:non-degenerate} is the least
common multiple of the orders of all roots of unity appearing as
quotients of characteristic roots of $\uu$.  This number can be
computed in time polynomial in $\size{\uu}$ since determining whether
an algebraic number $\alpha$ is a root of unity (and computing the
order of the root) can be done in polynomial time in
$\size{\alpha}$~\cite{TUCS05}.  From the representation (\ref{eq:exp}) we
see that if the original LRS is simple with characteristic roots
$\lambda_1,\ldots,\lambda_k$, then each subsequence $\langle u_{Mn+l}
\rangle_{n=0}^\infty$ is also simple, with characteristic roots
among $\lambda_1^M,\ldots,\lambda_k^M$.

The following is a celebrated result on LRS~\cite{Lec53,Mah35,Sko34}.
\begin{theorem}[Skolem-Mahler-Lech]
The set $\{ n : u_n=0\}$ of zeros of an LRS $\uu$ comprises a finite
set together with a finite number of arithmetic progressions.  If
$\uu$ is non-degenerate and not
identically zero, then its set of zeros is finite.
\label{thm:SML}
\end{theorem}

Suppose that $\uu$ and $\vv$ are LRS of orders $k$ and $l$
respectively, then the pointwise sum $\langle u_n+v_n
\rangle_{n=0}^\infty$ is an LRS of order at most $k+l$, and the
pointwise product $\langle u_nv_n \rangle_{n=0}^\infty$ is an LRS of
order at most $kl$.  Given representations of $\uu$ and $\vv$ we can
compute representations of the sum and product in polynomial time by
straightforward linear algebra.

\paragraph{First-Order Theory of the Reals.}
Let $\boldsymbol{x} = x_1, \ldots, x_m$ be a list of $m$ real-valued
variables, and let $\sigma(\boldsymbol{x})$ be a Boolean combination
of atomic predicates of the form $g(\boldsymbol{x}) \sim 0$, where
each $g(\boldsymbol{x})$ is a polynomial with integer coefficients in
the variables $\boldsymbol{x}$, and $\sim$ is either $>$ or $=$.  We
consider the problem of deciding the truth over the field $\mathbb{R}$
of sentences $\varphi$ in the form
\begin{equation}
\label{for-formula}
Q_1 x_1 \ldots Q_m x_m \, \sigma (\boldsymbol{x}) \, ,
\end{equation}
where each $Q_i$ is one of the quantifiers $\exists$ or $\forall$.
We write $\size{\varphi}$ for the
length of the syntactic representation of $\varphi$.

The collection of true sentences of the form (\ref{for-formula}) is
called the \textbf{first-order theory of the reals}.  Tarski famously showed
that this theory admits quantifier elimination and is therefore
decidable.  In this paper we rely on decision procedures for two
fragments of this theory.  We use the result of Canny~\cite{Canny88}
that if each $Q_i$ is a universal quantifier, then the truth of
$\varphi$ can be decided in space polynomial in $\size{\varphi}$.  We
also use the result of Renegar~\cite{Ren92} that for each fixed $M \in
\mathbb{N}$, if the number of variables in $\varphi$ is at most $M$,
then the truth of $\varphi$ can be determined in time polynomial in
$\size{\varphi}$.

Given a representation of an algebraic number $\alpha$, as described
in Section~\ref{sec:background}, both the real and imaginary parts of
$\alpha$ are straightforwardly definable by quantifier-free formulas
$\varphi(x)$ of size polynomial in $\size{\alpha}$.

\section{Multiplicative Relations}
Throughout this section let
$\boldsymbol{\lambda}=(\lambda_1,\ldots,\lambda_s)$ be a tuple of
algebraic numbers, each of height at most $H$ and degree at most $d$.
Assume that each $\lambda_i$ is represented in the manner described in
Section~\ref{sec:background}.  

We define the group of multiplicative
relations holding among the $\lambda_i$ to be the subgroup
$L(\boldsymbol{\lambda})$ of $\mathbb{Z}^s$ defined by
\[ L(\boldsymbol{\lambda}) = \{ (v_1,\ldots,v_s) \in \mathbb{Z}^s :
\lambda_1^{v_1} \ldots \lambda_s^{v_s} = 1\} \, . \] 

Bounds on the complexity of computing a basis of
$L(\boldsymbol{\lambda})$, considered as a free abelian group,
can be obtained from the following result of Masser~\cite{Mas88}
which gives an upper bound on the magnitude of the entries of the
vectors in such a basis.

\begin{theorem}[Masser]
The free abelian group $L(\boldsymbol{\lambda})$ has a basis
$\boldsymbol{v}_1,\ldots,\boldsymbol{v}_l \in \mathbb{Z}^s$ for which
\[\max_{1 \leq i \leq l,\, 1 \leq j \leq s} |\boldsymbol{v}_{i,j}| 
= (d \log H)^{O(s^2)}\, . \]
\label{thm:masser}
\end{theorem}

\begin{corollary}
A basis of $L(\boldsymbol{\lambda})$ can be computed in space
polynomial in $\size{\boldsymbol{\lambda}}$.  If $s$ and $d$ are
fixed, such a basis can be computed in time polynomial in
$\size{\boldsymbol{\lambda}}$.
\label{corl:basis}
\end{corollary}
\begin{proof}
Masser's bound entails that there is a basis
$\boldsymbol{v}_1,\ldots,\boldsymbol{v}_l$ whose total bit length is polynomial in
$s$, $\log d$ and $\log\log H$, all of which are
polynomial in $\size{\boldsymbol{\lambda}}$.  Moreover the membership
problem ``$\lambda_1^{v_1}\ldots \lambda_s^{v_s}=1$?''  for a
potential basis vector $\boldsymbol{v} \in \mathbb{Z}^s$ is decidable
in space polynomial in $\size{\boldsymbol{\lambda}}$ by reduction to
the decision problem for existential sentences over the reals.  

A set of vectors $\boldsymbol{v}_1,\ldots,\boldsymbol{v}_l$ in
$L(\boldsymbol{\lambda})$ is a basis if every vector $\boldsymbol{v}
\in L(\boldsymbol{\lambda})$ whose entries satisfy the bound in
Theorem~\ref{thm:masser} lies in the integer span of
$\boldsymbol{v}_1,\ldots,\boldsymbol{v}_l$.  For each such vector
$\boldsymbol{v}$ this can be checked by solving a system of linear
equations over the integers.  Thus we can compute a basis of
$L(\boldsymbol{\lambda})$ in space polynomial in
$\size{\boldsymbol{\lambda}}$ by brute-force search.

If $s$ and $d$ are fixed then the same brute-force search can be done
in time polynomial in $\size{\boldsymbol{\lambda}}$, noting that the
number of possible bases is polynomial in
$\size{\boldsymbol{\lambda}}$ and the membership problem
``$\lambda_1^{v_1}\ldots \lambda_s^{v_s}=1$?'' is decidable in
time polynomial in $\size{\boldsymbol{\lambda}}$ by reduction to the
decision problem for existential sentences over the reals with a fixed
number of variables.\qed
\end{proof}

The following is an easy consequence of Corollary~\ref{corl:basis}.
\begin{corollary}
Given $M \in \mathbb{N}$, a basis of $L(\lambda_1^M,\ldots,\lambda_s^M)$
can be computed in space polynomial in $\size{M}$ and
$\size{\boldsymbol{\lambda}}$.
\label{corl:scale}
\end{corollary}

Next we relate the group $L(\boldsymbol{\lambda})$ to the
\defemph{orbit} $ \{(\lambda_1^n,\ldots,\lambda_s^n) \mid n \in
\mathbb{N}\}$ of $\boldsymbol{\lambda}$.  Recall
from~\cite{Cassels65} the following classical theorem of Kronecker
on inhomogeneous Diophantine approximation.

\begin{theorem}[Kronecker]
Let $\theta_1,\ldots,\theta_s$ and $\psi_1,\ldots,\psi_s$ be real
numbers.  Suppose moreover that for all integers $u_1,\ldots,u_s$, if
$u_1\theta_1 + \ldots + u_s \theta_s \in \mathbb{Z}$ then also
$u_1\psi_1 + \ldots + u_s \psi_s \in \mathbb{Z}$, i.e., all integer
relations among the $\theta_i$ also hold among the $\psi_i$ (modulo
$\mathbb{Z}$).  Then for each $\varepsilon > 0$, there exist integers
  $p_1,\ldots,p_s$ and a non-negative integer $n$ such that
$|n\theta_i-p_i-\psi_i| \leq \varepsilon$.
\label{thm:kronecker}
\end{theorem}

Write $\mathbb{T} = \{ z \in \mathbb{C} : |z|=1\}$ and
consider the $s$-dimensional torus $\mathbb{T}^s$ as a group under
coordinatewise multiplication.  The following can be seen as a
multiplicative formulation of Kronecker's Theorem.

\begin{proposition}
  Let $\boldsymbol{\lambda}=(\lambda_1,\ldots,\lambda_s) \in
  \mathbb{T}^s$ and consider the group $L(\boldsymbol{\lambda})$ of
  multiplicative relations among the $\lambda_i$.  Define a subgroup
  $T(\boldsymbol{\lambda})$ of the torus $\mathbb{T}^s$ by
  \[ T(\boldsymbol{\lambda}) = \{ (\mu_1,\ldots,\mu_s) \in \mathbb{T}^s \mid
  \mu_1^{v_1} \ldots \mu_s^{v_s}=1 \mbox{ for all }
  \boldsymbol{v} \in L(\boldsymbol{\lambda}) \} \, . \] Then the
  orbit $S=\{(\lambda_1^n,\ldots,\lambda_s^n) \mid n \in \mathbb{N}\}$
  is a dense subset of $T(\boldsymbol{\lambda})$.
\label{prop:density}
\end{proposition}

\begin{proof}
For $j=1,\ldots,s$, let $\theta_j \in \mathbb{R}$ be such that
$\lambda_j = e^{2\pi i \theta_j}$.  Notice that multiplicative
relations $\lambda_1^{v_1}\ldots \lambda_s^{v_s}=1$ are in one-to-one
correspondence with additive relations $\theta_1v_1+\ldots+\theta_sv_s
\in \mathbb{Z}$.  Let $(\mu_1,\ldots,\mu_s)$ be an arbitrary element
of $T(\boldsymbol{\lambda})$, with $\mu_j = e^{2\pi i\psi_j}$ for some
$\psi_j\in \mathbb{R}$.  Then the hypotheses of
Theorem~\ref{thm:kronecker} apply to $\theta_1,\ldots,\theta_s$ and
$\psi_1,\ldots,\psi_s$.  Thus given $\varepsilon>0$, there exist $n
\geq 0$ and $p_1,\ldots,p_s \in \mathbb{Z}$ such that 
$|n\theta_j-p_j-\psi_j| \leq \varepsilon$
for
$j=1,\ldots,s$.  Whence for $j=1,\ldots,s$,
\[ |\lambda^n_j - \mu_j | = |e^{2\pi i (n\theta_j-p_j)}- e^{2\pi i\psi_j}| \leq
|2\pi (n\theta_j - p_j -\psi_j)| \leq 2\pi\varepsilon \, . \] It follows that $(\mu_1,\ldots,\mu_s)$
lies in the closure of $S$. \qed
\end{proof}

\section{Algorithm for Ultimate Positivity}
Let $K$ be a number field of degree $d$ over $\mathbb{Q}$.  Recall
that there are $d$ distinct field monomorphisms
$\sigma_1,\ldots,\sigma_d : K \rightarrow \mathbb{C}$ (see,
e.g.,~\cite{FT93}).  Given a finite set $S$ of prime ideals in the
ring of integers $\mathcal{O}$ of $K$, we say that $\alpha \in
\mathcal{O}$ is an $S$-\defemph{unit} if the principal ideal
$(\alpha)$ is a product of prime ideals in $S$.  The following lower
bound on the magnitude of sums of $S$-units, whose key ingredient is
Schlickewei's ${p}$-adic generalisation of Schmidt's Subspace Theorem,
was established in~\cite{Evertse84,PS82} to analyse the growth of
LRS\@.
\begin{theorem}[Evertse, van der Poorten, Schlickewei]
  Let $m$ be a positive integer and $S$ a finite set of prime
  ideals in $\mathcal{O}$.  Then for every $\varepsilon > 0$ there
  exists a constant $C$, depending only on $m$, $K$, $S$, and
  $\varepsilon$ with the following property:  for any set of $S$-units
  $x_1,\ldots,x_m \in \mathcal{O}$ such that $\sum_{i \in I} x_i \neq
  0$ for all non-empty $I \subseteq \{1,\ldots,m\}$, it holds that
\begin{gather}
|x_1 + \ldots + x_m| \geq C X Y^{-\varepsilon} \, ,
\end{gather}
where $X=\displaystyle\max\{|x_i|:{1 \leq i \leq m}\}$,
$Y=\max \{ |\sigma_j(x_i)| :
1 \leq i \leq m,\, 1 \leq j \leq d \}$.
\label{thm:sunit}
\end{theorem}

We first consider how to decide Ultimate Positivity in the case of a
non-degenerate simple LRS $\uu$.  As explained in
Section~\ref{sec:background}, we can assume without loss of generality that
$\uu$ has a positive real dominant root.  Furthermore, by considering
the LRS $\langle k^{n+1} u_n \rangle_{n=0}^\infty$ for a suitable integer
$k\geq 1$, we may assume that the characteristic roots and
coefficients in the closed-form solution (\ref{eq:exp}) are all
algebraic integers.  

Suppose that $\uu$ has dominant characteristic roots
$\rho,\gamma_1,\overline{\gamma_1},\ldots,\gamma_s,\overline{\gamma_s}$,
where $\rho$ is real and positive.  Then we can write $\uu$ in the
form
\begin{equation}
\label{eq-explanation}
u_n = b\rho^n + c_1\gamma_1^n + \overline{c_1} \overline{\gamma_1}^n + \ldots
+ c_s\gamma_s^n + \overline{c_s} \overline{\gamma_s}^n + r(n) \, ,
\end{equation}
where $r(n) = o(\rho^{n(1-\varepsilon)})$ for some $\varepsilon>0$.
Now let $\lambda_i = \gamma_i/\rho$ for $i=1,\ldots,s$.  Then we can write 
\begin{equation}
\label{eq-function}
u_n = \rho^n f(\lambda_1^n,\ldots,\lambda_s^n) + r(n) \, ,
\end{equation}
where $f : \mathbb{T}^s \rightarrow \mathbb{R}$ is defined by
$f(z_1,\ldots,z_s) = b+c_1z_1+\overline{c_1}\overline{z_1} + \ldots
+ c_s z_s + \overline{c_s} \overline{z_s}$.

\begin{proposition}
The LRS $\langle u_n \rangle_{n=0}^{\infty}$ is ultimately positive if
and only if $f(\boldsymbol{z}) \geq 0$ for all ${\boldsymbol{z} \in
  T(\boldsymbol{\lambda})}$.
\label{prop:non-deg-case}
\end{proposition}
\begin{proof}
Consider the expression (\ref{eq-explanation}).  Let $K$ be the number
field generated over $\mathbb{Q}$ by the characteristic roots of $\uu$
and let $S$ be the set of prime ideal divisors of the dominant
characteristic roots
$\rho,\gamma_1,\overline{\gamma_1},\ldots,\gamma_s,\overline{\gamma_s}$
and the associated coefficients
$b,c_1,\overline{c_1},\ldots,c_s,\overline{c_s}$.  (These coefficients
lie in $K$ by straightforward linear algebra.)  
Then the term
\begin{gather}
b\rho^n + c_1\gamma_1^n + \overline{c_1} \overline{\gamma_1}^n +
\ldots + c_s\gamma_s^n + \overline{c_s} \overline{\gamma_s}^n
\label{eq:S-units}
\end{gather}
is a sum of $S$-units.  

Applying Theorem~\ref{thm:sunit} to the sum of $S$-units in
(\ref{eq:S-units}), we have $X=C_1\rho^n$ for some constant $C_1>0$ and
$Y = C_2\rho^n$ for some constant $C_2>0$ (since an embedding of $K$
into $\mathbb{C}$ maps characteristic roots to characteristic roots).
The theorem tells us that for each $\varepsilon>0$ there is a
constant $C>0$ such that
\[ |b\rho^n + c_1\gamma_1^n + \overline{c_1} \overline{\gamma_1}^n + \ldots
+ c_s\gamma_s^n + \overline{c_s} \overline{\gamma_s}^n| \geq
C\rho^{n(1-\varepsilon)} \] for all but finitely many values of $n$.
(Since $\uu$ is non-degenerate, it follows from the Skolem-Mahler-Lech
Theorem that each non-empty sub-sum of the left-hand side vanishes for
finitely many $n$.)

Now choose $\varepsilon>0$ such that $r(n)=o(\rho^{n(1-\varepsilon)})$
in (\ref{eq-explanation}).  Then for all sufficiently large $n$, $u_n
\geq 0$ if and only if $b\rho^n + c_1\gamma_1^n + \overline{c_1}
\overline{\gamma_1}^n + \ldots + c_s\gamma_s^n + \overline{c_s}
\overline{\gamma_s}^n > 0$.  Equivalently, looking at
(\ref{eq-function}), for all sufficiently large $n$ we have
$u_n \geq 0$ if and only if $ f(\lambda_1^n,\ldots,\lambda_s^n) \geq
0$.  But the orbit $\{ (\lambda_1^n,\ldots,\lambda_s^n) : n \in
\mathbb{N}\}$ is a dense subset of $T(\boldsymbol{\lambda})$ by
Proposition~\ref{prop:density}.  Thus $u_n$ is ultimately positive if
and only if $f(\boldsymbol{z}) \geq 0$ for all ${\boldsymbol{z} \in
  T(\boldsymbol{\lambda})}$.  \qed
\end{proof}

We can now state and prove our main result.
\begin{theorem}
The Ultimate Positivity Problem for simple LRS is decidable in
polynomial space in general, and in polynomial time for LRS of fixed order.
\end{theorem}
\begin{proof}
A decision procedure is given in the table below.
Correctness follows from the fact that $\uu$ is ultimately positive if and
only if each of the non-degenerate subsequences $\vv$ considered in
Step 2 is ultimately positive.  But ultimate positivity of these
subsequences is determined in Step 2.4 using
Proposition~\ref{prop:non-deg-case}.
It remains to account for the complexity of each step.

As noted in Section~\ref{sec:background}, Step 1 requires time
polynomial in $\size{\uu}$.

For LRS of fixed order, there is an absolute bound on $M$ in Step 2,
while for LRS of arbitrary order, $M$ is exponentially bounded in
$\size{\uu}$ by Proposition~\ref{prop:non-degenerate}.  We show that for
each subsequence $\vv$, Steps 2.1--2.4 require polynomial time
for fixed-order LRS and polynomial space in general.

Using iterated squaring, the coefficients $b_i$ and $c_j$ in the
expression (\ref{eq:closed2}) for $\vv$ are definable in terms of the
characteristic roots of $\uu$ and the corresponding coefficients in
the closed-form expression for $\uu$ by a polynomial-size first-order
formula that uses only universal quantifiers.  This accomplishes Step
2.1.

Combining Corollaries~\ref{corl:basis} and~\ref{corl:scale}, Step 2.3
can be done in polynomial space for arbitrary LRS and polynomial time
for LRS of fixed order.

Step 2.4 uses a decision procedure for universal sentences over the
reals, having already noted that the coefficients $b_i$ and $c_j$ are
first-order definable.  By the results described in
Section~\ref{sec:background} this can be done in polynomial space for
arbitrary LRS and polynomial time for LRS of fixed order.\qed
\end{proof}
\begin{table}[h]
\begin{tabular}{l}
\textbf{Decision procedure for ultimate positivity of a simple LRS $\uu$} \\
\hline
\multicolumn{1}{c}{\parbox{0.95\textwidth}{
\begin{itemize}
\item[1.] Compute the characteristic roots
  $\{\rho_1,\ldots,\rho_\ell,\gamma_1,\overline{\gamma_1},\ldots,
  \gamma_m,\overline{\gamma_m}\}$ of $\uu$.\\  Writing $\alpha\sim\beta$
  if $\alpha/\beta$ is a root of unity,
let $M = \mathrm{lcm} \{ \mathrm{ord}(\alpha/\beta) :
                         \mbox{$\alpha \sim \beta$ are characteristic roots}\}$.
Moreover let $\{ \rho_i : i \in I\} \cup \{ \gamma_j,\overline{\gamma_j} : j \in J\}$ contain a unique representative from each equivalence class.

\item[2.] For $l=0,\ldots,M-1$, check ultimate positivity of the
  non-degenerate subsequence $v_n =
  u_{Mn+l}$ as follows:
\begin{itemize}
\item[2.1.]
Compute the coefficients $b_i$ and $c_j$ in the closed-form solution
\begin{equation}
 v_n = \sum_{i\in I} b_i \rho_i^{Mn} + 
      \sum_{j \in J} \left(c_j \gamma_j^{Mn} + 
                  \overline{c_j} \overline{\gamma_j}^{Mn}\right) \, . 
\label{eq:closed2}
\end{equation}
\item[2.2.] If $\vv \not\equiv 0$ and there is no dominant real
  characteristic root in (\ref{eq:closed2}) then $\vv$ is not
  ultimately positive.
\item[2.3.] Let
  $\rho_1,\gamma_1,\overline{\gamma_1},\ldots,\gamma_s,\overline{\gamma_s}$
  be dominant among the characteristic roots appearing in
  (\ref{eq:closed2}).  Define
  $\lambda_1=\gamma_1/\rho_1,\ldots,\lambda_s=\gamma_s/\rho_1$ and
  compute a basis of $L(\lambda_1^M,\ldots,\lambda_s^M)$.
\item[2.4.] Define $f:\mathbb{T}^s\rightarrow \mathbb{R}$ by
$f(z_1,\ldots,z_s) = b_1+c_1z_1+\overline{c_1}
\overline{z_1} + \ldots + c_sz_s +
\overline{c_s} \overline{z_s}$.\\ 
Then $\vv$ is ultimately positive if and only if 
$f(\boldsymbol{z}) \geq 0$ for all
$\boldsymbol{z} \in T(\boldsymbol{\lambda}^M)$.
\end{itemize}
\end{itemize}
}}\\
\hline
\end{tabular}
\end{table}

We note that a related proof strategy (passing from a finitely
generated group to its closure and appealing to the theory of the
reals) was used in~\cite{BlondelJKP05} in the context of threshold
problems for quantum automata.

\section{Complexity Lower Bound}

In this section we give reductions of the decision problem for
universal sentences over the field of real numbers to the Positivity
and Ultimate Positivity Problems respectively.  The former problem is
easily seen to be \textbf{coNP}-hard and, through the work of
Canny~\cite{Canny88}, is known to be in \textbf{PSPACE}.  Typically
this \textbf{PSPACE} upper bound is stated for the complement problem:
the decision problem for existential sentences over the field of
reals.

It is known that the problem 4-FEAS of whether a degree-4 polynomial
has a real root is polynomial-time equivalent to the decision problem
for the existential theory of the reals~\cite{BSS}.  Here we consider
a related problem, 4-POS, which asks whether a degree-4 polynomial
$f(x_1,\ldots,x_n)$ with rational coefficients satisfies
$f(\boldsymbol{x})\geq 0$ for all $\boldsymbol{x} \in [0,1]^n$.  Using
the above-mentioned result on 4-FEAS in tandem with bounds on
magnitude of definable numbers in the existential theory of the reals
(see~\cite{OuaknineW13b} for details) we can show:

\begin{theorem}
There is a polynomial-time reduction of the decision problem for the
universal theory of the reals to the problem 4-POS\@.
\label{thm:reduce}
\end{theorem}

We now reduce 4-POS to the Positivity and Ultimate Positivity
Problems.  The first step of the reduction is to compute a collection
of $s$ multiplicatively independent algebraic numbers of absolute
value $1$.

By a classical result of Lagrange, a prime number is congruent to $1$
modulo $4$ if and only if it can be written as the sum of two
squares~\cite{FT93}.  By Theorem~\ref{thm:dirichlet}, the class of
such primes has asymptotic density $1/2$ in the set of all primes, and
therefore, by the Prime Number Theorem, asymptotic density ${1}/(2\log
n)$ in the set of natural numbers.  It follows that one can compute
the first $s$ such primes $p_1,\ldots,p_s$ and their decomposition as
sums of squares in time polynomial in $s$.  Writing $p_j = a_j^2 +
b_j^2$, where $a_j,b_j \in \mathbb{Z}$, define $\lambda_j =
\frac{a_j+ib_j}{a_j-ib_j}$ for $j=1,\ldots,s$.  Then each $\lambda_j$
  is an algebraic number of degree $2$ and absolute value $1$.

\begin{proposition}
$\lambda_1,\ldots,\lambda_s$ are multiplicatively independent.
\label{prop:independent}
\end{proposition}
\begin{proof}
Recall that the ring of Gaussian integers $\mathbb{Z}(i)$ is a unique
factorisation domain and that $a+ib \in \mathbb{Z}(i)$ is prime iff
$a^2+b^2$ is a rational prime~\cite{FT93}.  Now $\lambda_1^{n_1}\ldots
\lambda_s^{n_s}=1$ if and only if
\begin{equation*}
(a_1+ib_1)^{n_1}\ldots (a_s+ib_s)^{n_s} = 
(a_1-ib_1)^{n_1}\ldots (a_s-ib_s)^{n_s} 
\end{equation*}
But each factor $a_j+ib_j$ and $a_j-ib_j$ is prime by construction.
Thus by unique factorisation we must have $n_1=0,\ldots,n_s=0$.\qed
\end{proof} 

\begin{theorem}
There are polynomial-time reductions from 4-POS to the Positivity and
Ultimate Positivity Problems for LRS\@.
\end{theorem}
\begin{proof}
Suppose we are given an instance of 4-POS, consisting of a
polynomial $f(x_1,\ldots,x_s)$.  Let $\lambda_1,\ldots,\lambda_s$ be
multiplicatively independent algebraic numbers, constructed as in
Proposition~\ref{prop:independent}.  For $j=1,\ldots,s$, the sequence
$\langle y_{j,n} : n \in \mathbb{N} \rangle$ defined by $y_{j,n}
= \frac{1}{2}(\lambda^n_j+\overline{\lambda_j}^n)$ satisfies a
second-order linear recurrence $y_{j,n+2} = ({2a_j}/{p_j})y_{j,n+1} -
y_{j,n}$ with rational coefficients.

Recall, moreover, that given two simple LRS of respective orders $l$
and $m$, their sum is a simple LRS of order at most $l+m$, their product is a
simple LRS of order at most $lm$, and representations of both can be computed
in polynomial time in the size of the input LRS\@.  Thus the sequence
$\boldsymbol{u}=\langle u_n : n \in \mathbb{N} \rangle$ given by $u_n
= f(y^2_{1,n},\ldots,y_{s,n}^2)$ is a simple LRS over the rationals.
Since $f$ has degree at most $4$, the order of $\boldsymbol{u}$ is at
most $4^{4}$ times the number of monomials in $f$ and the recurrence
satisfied by $\boldsymbol{u}$ can be computed in time polynomial in
$\size{f}$.  (Observe that if the degree of $f$ were not fixed, then
the above reasoning would yield an upper bound on the order of
$\boldsymbol{u}$ that is exponential in the degree of $f$.)

From Propositions~\ref{prop:density} and~\ref{prop:independent} it
follows that the orbit $\{ (\lambda_1^n,\ldots,\lambda_s^n) : n \in
\mathbb{N} \}$ is dense in the torus $\mathbb{T}^s$.  Thus the set
$\{(y^2_{1,n},\ldots,y_{s,n}^2) : n \in \mathbb{N}\}$ is dense in
$[0,1]^s$ and $f$ assumes a strictly negative value on $[0,1]^s$ if
and only if $u_n<0$ for some (equivalently infinitely many) $n$.  This
completes the reduction.\qed
\end{proof}

\section{Conclusion}

We have shown that the Ultimate Positivity Problem for simple LRS is
decidable in polynomial space and as hard as the decision problem for
universal sentences over the field of real numbers.  A more careful
accounting of the complexity of our decision procedure places it in
$\mathbf{coNP}$ with an oracle for the universal theory of the reals.
Thus a \textbf{PSPACE}-hardness result for Ultimate Positivity would
have non-trivial consequences for the complexity of decision problems
for first-order logic over the reals.  On the other hand, the obstacle
to improving the polynomial-space upper bound is the complexity of
computing a basis of the group of multiplicative relations among the
characteristic roots of the recurrence.


\bibliographystyle{abbrv}
\bibliography{ultimate}
\end{document}